\documentclass{article}
\usepackage{ijcai23}

\usepackage{times}
\usepackage{soul}
\usepackage{url}
\usepackage[hidelinks]{hyperref}
\usepackage[utf8]{inputenc}
\usepackage[small]{caption}
\usepackage{graphicx}
\usepackage{amsmath}
\usepackage{amsthm}
\usepackage{booktabs}
\usepackage{algorithm}
\usepackage{algorithmic}
\usepackage[switch]{lineno}

\usepackage{amssymb}
\usepackage{bbm}
\usepackage{makecell}
\usepackage{cases}
\usepackage{color}
\usepackage{multirow}
\usepackage{multicol}
\usepackage{setspace}

\newtheorem{theorem}{Theorem}

\newtheorem{rem}{Remark}
\theoremstyle{definition}
\newtheorem{defn}{Definition}

\title{CSGCL: Community-Strength-Enhanced Graph Contrastive Learning}

\author{
Han Chen$^{1,2,*}$
\and
Ziwen Zhao$^{1,*}$\and
Yuhua Li$^{1,\dagger}$\and
Yixiong Zou$^{1}$\and
Ruixuan Li$^{1}$\And
Rui Zhang$^{3,\dagger}$
\affiliations
$^1$School of Computer Science and Technology, Huazhong University of Science and Technology\\
$^2$Institute of Artificial Intelligence, Huazhong University of Science and Technology\\
$^3$www.ruizhang.info\\
\emails
\{HanChenHUST, zwzhao, idcliyuhua, yixiongz, rxli\}@hust.edu.cn,
rayteam@yeah.net
}
\begin{document}

\maketitle
\renewcommand{\thefootnote}{\fnsymbol{footnote}}
\footnotetext[1]{Equal contributors.}
\footnotetext[2]{Corresponding authors.}

\begin{abstract}
    Graph Contrastive Learning (GCL) is an effective way to learn generalized graph representations in a self-supervised manner, 
    and has grown rapidly in recent years. 
    However, the underlying community semantics has not been well explored by most previous GCL methods. 
    Research that attempts to leverage communities in GCL regards them as having the same influence on the graph, 
    leading to extra representation biases. 
    To tackle this issue, we define ``community strength'' to measure the difference of influence among communities. 
    Under this premise, we propose a {\bf C}ommunity-{\bf S}trength-enhanced {\bf G}raph {\bf C}ontrastive {\bf L}earning (CSGCL) framework 
    to preserve community strength throughout the learning process. 
    Firstly, we present two novel graph augmentation methods, Communal Attribute Voting (CAV) and Communal Edge Dropping (CED), 
    where the perturbations of node attributes and edges are guided by community strength. 
    Secondly, we propose a dynamic ``Team-up'' contrastive learning scheme, 
    where community strength is used to progressively fine-tune the contrastive objective. 
    We report extensive experiment results on three downstream tasks: node classification, node clustering, and link prediction. 
    CSGCL achieves state-of-the-art performance compared with other GCL methods, 
    validating that community strength brings effectiveness and generality to graph representations.
    Our code is available at 
    \url{https://github.com/HanChen-HUST/CSGCL}.
\end{abstract}

\section{Introduction}

\begin{figure}[ht]
  \centering
  \includegraphics[scale=0.42]{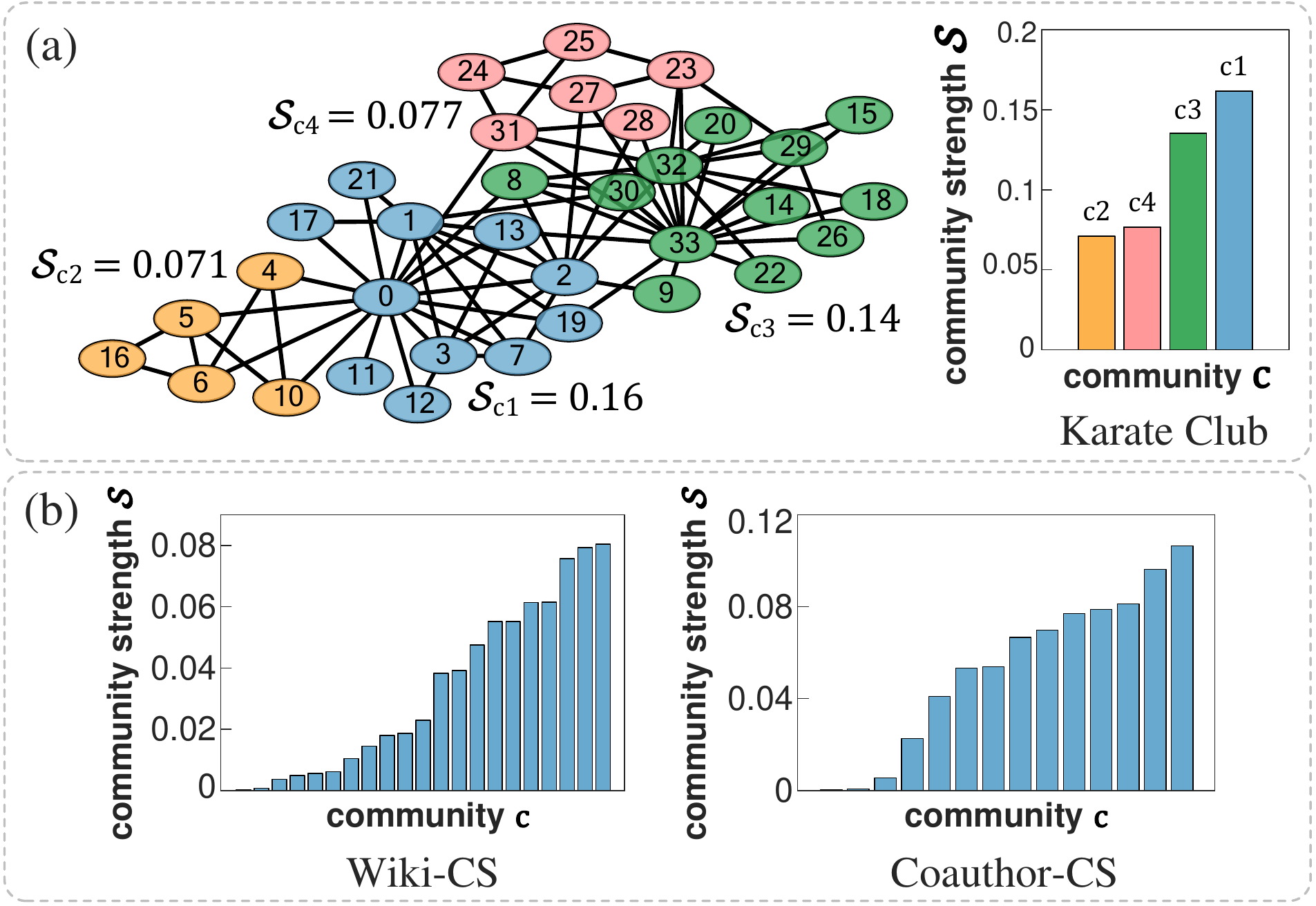}
  \caption{
  % The power of community strength. 
  Community strength varies in real-world networks.
  (a) Community strength on the Karate Club dataset~\protect\cite{karateclub}. 
  Communities, marked in different colors, are subgraphs with dense internal links.
  (b) 
  Community strength distributions of two real-world network examples:
  Wiki-CS~\protect\cite{karateclub} 
  and Coauthor-CS~\protect\cite{amazon+coauthor},
  which suggest that there are differences in community strength in real-world networks.
  }
  \label{fig-cs2}
\end{figure}

Graph Representation Learning (GRL) is extensively used in 
% both industrial and academic fields, such as 
knowledge graphs~\cite{kg}, recommendation~\cite{lightgcn}, e-commerce~\cite{li2020}, biological systems~\cite{cosmig}, etc. 
Unlike the label-dependent and noise-sensitive supervised GRL methods,
% which are also hard to acquire for large-scale network data.
Graph Contrastive Learning (GCL)~\cite{cl} extracts well-generalized representations between paired data views in a self-supervised manner. 
Many promising results from node-level GCL studies have been achieved~\cite{dgi,mvgrl,gca}.
Despite their success, most of these studies hardly consider the community in GCL. 

% Of the multiple levels of graph hierarchy, 
A {\bf community} is a high-order structure of a graph,
characterized by a denser connection among its members~\cite{girvan-newman}.
% Unlike microscopic node representations, small motifs, or randomly partitioned subgraphs, communities are composed by node groups with stronger connectivity between nodes. 
This underlying structure is highly informative:
for citation networks, communities indicate different academic fields; 
for co-purchase networks, communities indicate the common interests of a customer in multiple types of merchandise.
% Many studies have also found that community is closely correlated with various downstream tasks on graphs~\cite{cpne},
Therefore, it is necessary to take into account communities in GCL.
Recently, research like gCooL~\cite{gcool} combines community detection with GCL. 
However, their work is built on an underlying hypothesis
that different communities have the same influence on global semantics,
which may lead to extra representation biases in GCL.

To tackle this issue, we hold that 
{\bf communities influence global semantics differently}, 
which is vital for GCL.
We define ``{\bf community strength}'' to measure this influence of communities,
and give a visualized example of community strength in Figure~\ref{fig-cs2}(a).
Community strength, denoted as $\mathcal{S}$, 
indicates {\bf how dense the connection of a community $c$ is}, 
and reflects the influence of $c$ on the entire graph.
A stricter definition can be found in Section~\ref{3.2}.
``Strong communities'', with more densely connected edges,
have a larger $\mathcal{S}$
({\it e.g.} $\mathcal{S}_{c1}$ and $\mathcal{S}_{c3}$); 
``weak communities'', more sparsely connected ones, 
have a smaller $\mathcal{S}$
({\it e.g.} $\mathcal{S}_{c2}$ and $\mathcal{S}_{c4}$).
% Hence, community strength is closely relative to the graph structure and semantics.
As shown in Figure~\ref{fig-cs2}(b), 
the difference of community strength is one of the inherent features of real-world networks.
Hence, community strength is of practical importance.
We naturally ask:
{\it is there an effective way to leverage community strength information in GCL?}

To answer this question, we propose a novel Community-Strength-enhanced Graph Contrastive Learning (CSGCL) framework. 
It can capture and preserve community strength information throughout the learning process, 
from data augmentation to the training objective.
Firstly, we put forward two novel graph augmentation approaches, Communal Attribute Voting (CAV) and Communal Edge Dropping (CED), 
which apply community strength respectively to every attribute and every edge to preserve the differences among communities.
Secondly, we propose a dynamic ``Team-up'' contrastive scheme for CSGCL
to progressively fine-tune the contrastive objective with community strength. 

\begin{figure}
  \centering
  \includegraphics[scale=0.69]{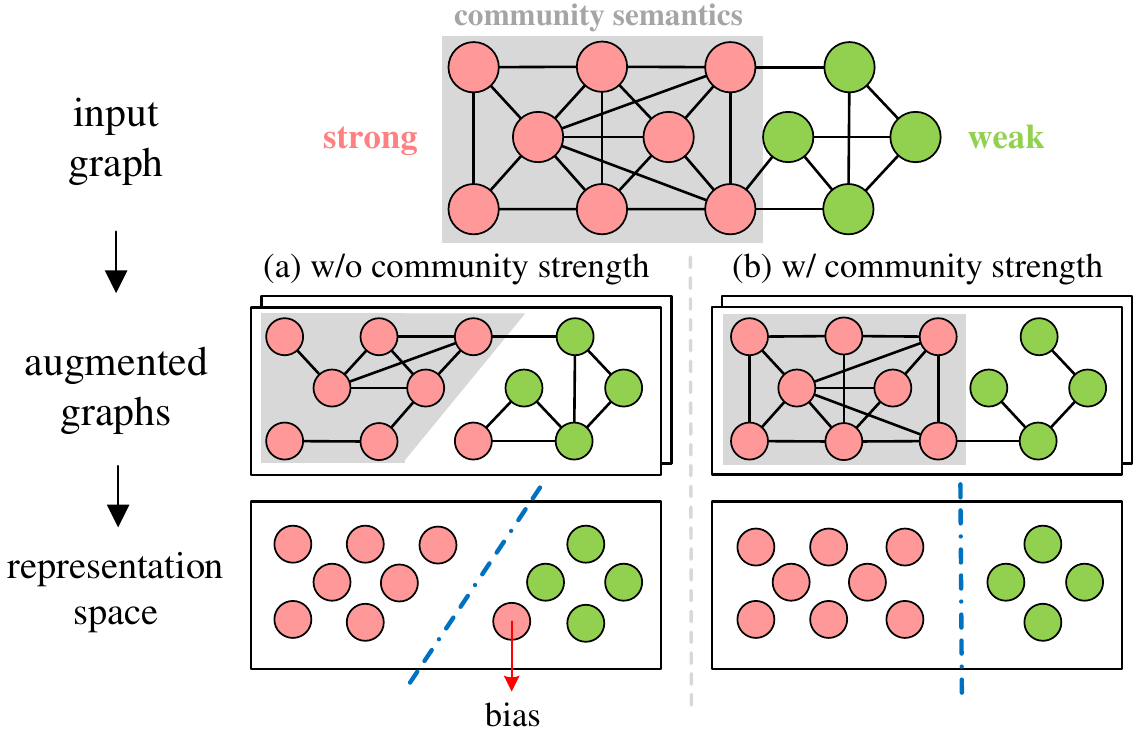}
  \caption{Illustration of the advantage of community strength to GCL.
  (a) Node representations without the guidance of community strength,
  where strong and weak communities are treated equally.
  Therefore, community semantics is changed during the augmentation,
  resulting in a bias in the representation space.  
  (b) Node representations by CSGCL,
  where the bias is handled by preserving community strength information.
  Note that here we use a schematic subgraph example from some dataset with more communities.
  }
  \label{fig-csisgood}
\end{figure}

With the help of community strength, CSGCL can learn more discriminative representations,
as in Figure~\ref{fig-csisgood}.
Without the consideration of community strength, every community is treated equally, 
which is not the case for real-world networks,
% causing the influential vulnerable to the perturbation,
% Thus, strong community semantics showed to the encoder may have already been damaged, 
so it leads to more representation biases.
CSGCL takes into account community strength to handle these biases.

Our main contributions are as follows:

\begin{itemize}
    \item We give a definition of community strength to evaluate the influence of a community. 
    To the best of our knowledge, we are the first to manage to preserve community strength information in GCL. 
    % Different from existing methods, we are the first to integrate community strength into the design of graph augmentations and contrastive objectives.
    % Our method outperforms state-of-the-art significantly on many tasks, indicating the effectiveness and versatility of community strength.
    \item We propose two enhanced graph augmentation methods, CAV and CED, 
    in which the perturbations of node attributes and edges are guided by community strength. 
    % This preserves important community information from the beginning of contrastive learning. 
    \item A dynamic Team-up objective is devised, 
    which directs the training process to preserve community strength.
    % , beneficial community semantics is condensed for downstream use. 
    % This preserves important community information at the end of contrastive learning. 
    \item Extensive experiments on graph benchmark datasets show that 
    CSGCL achieves state-of-the-art performance on up to three downstream tasks -- 
    node classification, node clustering, and link prediction. 
    % which backs our argument of the importance of community strength. 
\end{itemize}

\section{Related Work}

In this section, 
we give a brief review of traditional GRL methods, node-level GCL methods, and GCL with communities. 
Then, we compare CSGCL with these methods. 
For additional related work, see Appendix~\ref{D}.

\subsection{Graph Representation Learning (GRL)}

GRL extracts and compresses the semantic information and topological structure of a graph into low-dimensional representation vectors. 
Traditional methods~\cite{deepwalk,node2vec} employ random walks to generate node sequences 
and embed them into the representation space using text encoding approaches.
% such as Skip-gram~\cite{skipgram}. 
These methods only consider the co-occurrence of nodes. 
As improvements, 
CPNE~\cite{cpne} preserves community information
by Non-negative Matrix Decomposition
and, WGCN~\cite{wgcn} captures weighted structure features.
% Subsequently, more GRL methods that focus on community preservation have emerged~\cite{cde,gne,aneci}. 
An effective GRL paradigm is graph self-supervised learning~\cite{gae+vgae,gssl}, 
% ~\cite{gae+vgae,graphmae,gssl}. 
which creates pretext tasks to leverage the information within data, 
helping the model to get rid of label dependency. 
However, 
% according to our knowledge, 
they do not preserve community strength well in their representations.

\begin{figure*}
  \centering
  \includegraphics[scale=0.92]{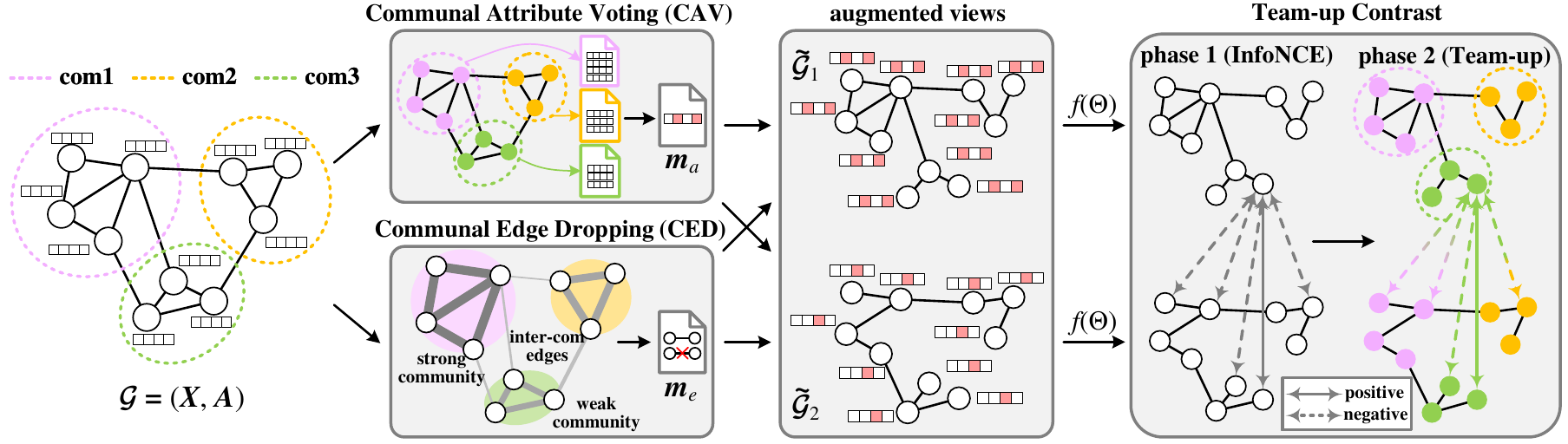}
  \caption{Overview of the proposed CSGCL framework.
  Communities of the input graph $\mathcal{G}$ are marked in three different colors.
  Firstly, two augmented views are generated from $\mathcal{G}$ 
  using two concurrent graph augmentations, CAV and CED:
  CAV tends to remove the most-voted-for attributes,
  where each community submits a ``vote'' (colored file icons) with attribute penalties;
  CED tends to retain edges inside communities, especially the strong ones (thick),
  and drop inter-community edges (thin).
  Then, two augmented views are fed into a shared encoder and projector $f(\Theta)$ to get the representation graphs.
  In the end, the Team-up contrastive scheme (phase 2) progressively fine-tunes the InfoNCE objective (phase 1).
  }
  \label{fig-csgcl}
\end{figure*}

\subsection{Graph Contrastive Learning (GCL)}

Shortly after the proposal of contrastive learning~\cite{dim,simclr}, 
% ~\cite{dim,moco,simclr}
it was introduced into the GRL field and became a new self-supervised hotspot.
Most existing GCL methods~\cite{grace,afgrl,rosa,cgi} only focus on node-level relations. 
GraphCL~\cite{graphcl} introduces augmentations of nodes or node pairs like DropEdge~\cite{dropedge} and feature masking into GCL to increase the difference between two graph views. 
% GRACE~\cite{grace} further amends the contrastive loss by considering all other nodes as negative pairs. 
GCA~\cite{gca} further improves graph augmentation methods to adaptively capture the impact of different nodes. 
% COSTA~\cite{costa} proposes a bias-free augmentation scheme on graph representations by generating discriminative graph sketches. 
% To mitigate the negative effect of augmentations to graph semantics,
% AFGRL~\cite{afgrl} uses the original graph as well as a sampled graph 
% with well-chosen positive samples as contrastive views to learn representations augmentation-freely.
Outstanding as they are, they \textbf{ignore the inherent community information in networks}. 
Other works which consider the global information of graphs~\cite{dgi,mvgrl} 
also pay little attention to community semantics, 
which is discarded or undermined in node-level readout operations. 

gCooL~\cite{gcool} makes a good attempt to combine community information with GCL. 
It partitions communities on the representation graph and contrasts node-level and community-level representations across two different views,
considering node pairs of the same community as positive samples. 
However, \textbf{gCooL does not capture the difference between strong and weak communities}.
% , so it has not done its best to mitigate graphic biases.
The differences between CSGCL and the aforementioned methods are that
(1) CSGCL is a contrastive method taking advantage of community strength information; 
(2) CSGCL preserves community strength from data augmentation to model optimization.
% (3) unlike gCooL, CSGCL obtains community information from input graph to avoid edge perturbation biases.

\begin{table}
\centering
\tabcolsep=0.25cm
\resizebox{0.85\columnwidth}{!}{
\begin{tabular}{cl} \toprule 
Symbol & Description \\ \midrule
$\mathcal{G}, \tilde{\mathcal{G}}$ & original \& perturbed attributed graph \\
$\mathcal{V}$ & node set of graph $\mathcal{G}$ \\
$\mathcal{E}$ & edge set of graph $\mathcal{G}$ \\
$\mathbf{X}, \tilde{\mathbf{X}}$ & original \& perturbed node attribute matrix \\
$\mathbf{A}, \tilde{\mathbf{A}}$ & original \& perturbed adjacency matrix \\
$\mathbf{Y}$ & classification labels of graph $\mathcal{G}$ \\
$\mathbf{Z}$ & representation matrix of graph $\mathcal{G}$ \\
$\mathcal{C}$ & community set of graph $\mathcal{G}$ \\ 
$\mathcal{E}_c$ & edge set of a community $c$ \\
$\mathcal{S}$ & vector of strength $\mathcal{S}_c$ of every community $c$\\
$\mathbf{H}$ & community indicator matrix \\
\midrule\midrule
$\mathbf{M}_{i:}$ & $i$th row of any matrix $\mathbf{M}$\\
$\mathbf{M}_{:j}$ & $j$th column of any matrix $\mathbf{M}$\\
$M_{i, j}$ & element at row $i$ and column $j$ of any matrix $\mathbf{M}$ \\
$\mathbf{M}^{(k)}$ & belonging to the $k$th view of any matrix $\mathbf{M}$, $k=1,2$ \\
$d(\cdot)$ & degree of a node \\
$f(\cdot;\Theta)$ & graph encoder with parameters $\Theta$ \\
$\mathcal{T}(\cdot)$ & data augmentations \\
$\mathbbm{1}_{[ex]}$ & indicator function, $\begin{cases}1, ex \ \text{is true}\\0, ex \ \text{is false}\end{cases}$ \\
$|\cdot|$ & size of a set \\
$\|\cdot\|$ & Euclid norm \\
$\circ$ & Hadamard product \\
\bottomrule
\end{tabular}
}
\caption{Notations.}
\label{table1}
\end{table}

\section{Community-Strength-enhanced Graph Contrastive Learning}

In this section, we illustrate and formulate the details of CSGCL. 
The overall architecture of CSGCL is shown in Figure~\ref{fig-csgcl}. 
Notations used below are summarized in Table~\ref{table1}.
For an intuitive workflow of CSGCL, we summarize our approaches in pseudo-codes in Appendix~\ref{A}.

\subsection{Graph Contrastive Learning}

An undirected attributed graph $\mathcal{G}=(\mathcal{V}, \mathcal{E})$ is composed of a vertex set $\mathcal{V}$ and an edge set $\mathcal{E}$. 
In a graph dataset, a graph with $n$ nodes and $d$ attributes is given as 
a node attribute matrix $\mathbf{X} \in \mathbb{R}^{n\times d}$ 
and a symmetric adjacency matrix $\mathbf{A} \in \{0, 1\}^{n \times n}$, 
denoted as $\mathcal{G}=(\mathbf{X}, \mathbf{A})$.
% The $i$th row of the attribute matrix $\mathbf{X}_{i:}$ is the attributes of the $i$th node.
For nodes $u,v \in \mathcal{V}$, $A_{(u,v)}=1$ iff edge $(u,v) \in \mathcal{E}$, otherwise $A_{(u,v)}=0$.
% $d(v)$ is used to denote the degree of node $v$, 
% which means the total number of edges connected with $v$.

Graph Contrastive Learning aims to maximize the agreement of pairwise graph embeddings~\cite{dim}.
In general, A GCL framework is comprised of a GNN encoder, a projection head, and a contrastive loss. 
At the beginning of training, 
GCL generates two graph views
$(\tilde{\mathcal{G}}_1, \tilde{\mathcal{G}}_2)$
by random perturbation to the input graph.
Such perturbation is usually from multiple data augmentations $\mathcal{T}(\mathcal{G})$, 
{\it e.g.} DropEdge~\cite{dropedge}
%, DropNode~\cite{graphcl}, 
and feature masking~\cite{grace}. 
Then, two views are both fed into a shared GNN encoder (like a two-layer GCN~\cite{gcn}) to obtain node representations 
(but only the encoder is retained for various downstream predictions).
Let $\mathbf{Z}=f(\tilde{\mathcal{G}}; \Theta)$ be the projected embedding of $\tilde{\mathcal{G}}$,
where $f$ is the graph encoder followed by a nonlinear projection layer. 
The embeddings of a certain node $i$ in different views $(\mathbf{Z}^{(1)}_{i:},\mathbf{Z}^{(2)}_{i:})$ are seen as positive pairs, 
and all other pairs (embeddings of different nodes in the same or different views) are seen as negative pairs. 

GCL usually employs noise contrastive estimation (NCE) as the optimization objective.
The most commonly used one is InfoNCE~\cite{cpc},
which enhances the discrimination of representations by pulling similar (positive) ones together and pushing dissimilar (negative) ones away.
The similarity between node pairs is calculated as
\begin{equation}\label{eq1}
s_{ij}^{(1,2)} = sim(\mathbf{Z}^{(1)}_{i:}, \mathbf{Z}^{(2)}_{j:}) / \tau,
\end{equation}
where $sim(\mathbf{z}_i, \mathbf{z}_j) = (\mathbf{z}_i^\top \mathbf{z}_j) / \left(\|{\mathbf{z}_i}\|\|{\mathbf{z}_j}\|\right),$ 
and $\tau$ is the temperature of the similarity. 
Then InfoNCE is formulated as 
\begin{linenomath}
\begin{align}\label{eq2}
    &\mathcal{L} = \mathbb{E}_{(\tilde{\mathcal{G}}_1, \tilde{\mathcal{G}}_2) \sim \mathcal{T}(\mathcal{G})} \nonumber \\
    &\left(\!-\frac{1}{n}\!\sum_{i=1}^{n}{
    \log{\displaystyle\frac{\exp(s_{ii}^{(1,2)})}
    {\sum_{\substack{j=1 \\ j \ne i}}^{n}\exp(s_{ij}^{(1,1)})\!+\!\sum_{j=1}^{n}\exp(s_{ij}^{(1,2)})}}
    }\!\right).
\end{align}
\end{linenomath}

\subsection{{\spaceskip=0.2em\relax  Community-strength-enhanced Augmentations}}\label{3.2}

In this section, we introduce the community-strength-enhanced augmentations:
Communal Attribute Voting and Communal Edge Dropping.
% First, we will give the formal definition of ``community strength''.

\paragraph{Community strength.} 
To begin with, communities in the graph are highlighted by unsupervised community detection methods. 
CSGCL utilizes existing community detectors
so that our study can refocus attention on the communities in graph representations. 
We give more details of community detection methods suitable for CSGCL in Appendix~\ref{D}.

Let $\mathcal{C}$ be the set of all communities in graph $\mathcal{G}$. 
% Community strength is determined by the internal quality of subgraphs, which is a particular partition of a network. 
% Therefore, as a measure of graph partition quality, we introduce modularity~\cite{mod} as follows: 
% \begin{equation}\label{eqmod}
%     \mathcal{M}=\frac{1}{2|\mathcal{E}|}\sum_{c \in \mathcal{C}}\left(2|\mathcal{E}_c|-\frac{(\sum_{v \in c} d(v))^2}{2|\mathcal{E}|}\right)
% \end{equation}
% where $|\mathcal{E}_c|$ denotes the total number of edges in community $c$. 
%$\mathcal{M}$ measures the discrepancy of edge density between intra-community and inter-community node groups~\cite{louvain}. 
%A good partition of communities is able to maximize the graph-level modularity.
%
% Now we consider every entry of the summation in \eqref{eqmod}. 
Then, we have the definition below:
\defn (Community strength $\mathcal{S}$)
For every community $c$ of an undirected attributed graph $\mathcal{G}=(\mathcal{V}, \mathcal{E})$, 
% each of its internal node $v \in c$ has a degree $d(v) > 0$. 
let $|\mathcal{E}|$ and $|\mathcal{E}_c|$ be the number of edges of the entire graph and inside that community respectively. 
Then, the community strength $\mathcal{S}_c$ of community $c$ is formulated as follows:
\begin{equation}\label{eq3}
    \mathcal{S}_c=\frac{|\mathcal{E}_{c}|}{|\mathcal{E}|} - \frac{(\sum_{v \in c} d(v))^2}{4{|\mathcal{E}|}^2}
\end{equation}
which is composed of two terms:
the proportion of edges inside community $c$, 
and the inter-community connection between $c$ and other communities.
% We give the intuitive meaning of $\mathcal{S}_c$ by rearranging eq (3) as (please note the underbraces)
% \begin{equation}\notag
% \mathcal{S}_c = \underbrace{|\mathcal{E}_{c}|/|\mathcal{E}|}_{
% \%\text{ of edges in } c}
% - (\underbrace{(\textstyle\sum_{v \in c}{d(v)})/2|\mathcal{E}_{c}| \cdot |\mathcal{E}_{c}|/|\mathcal{E}| }_{
% \text{inter-community connection density}
% })^2
% \end{equation}
% where the first term gives the proportion of edges inside community $c$, 
% \textit{i.e.} the relative size of $c$; 
% the second term reflects the inter-community connection between $c$ and other communities, 
% so $c$ is relatively dense connected if the second term is small. 
% The square operation is to balance the result of $\mathcal{S}_c$.
% We view a relatively larger community with denser internal connections more influential. 

Community strength is based on modularity~\cite{mod}, a measure of graph partition quality, 
which compares every community with its randomly connected counterpart, 
and calculates the difference in their number of edges. 
% The larger the difference, the more prominent the community is in the entire graph, which meets our need perfectly. 
However, 
community strength is a metric of the local topological structure,
while modularity is a quality metric of the entire graph.
$\mathcal{S} > 0$ for every community,
because $\mathcal{S}_c \le 0$ means $c$ does not have any community characteristics.
% {\it i.e.} its internal edge density is no more than that of a randomly linked subgraph~\cite{mod}.

% ``A dense community has a greater impact'' has been a mild assumption in graphs.
% Intuitively, dense community members have large degrees which are more correlated with other nodes, 
% thus having a greater impact on the overall situation.
% For example, in social networks where social influence can be regarded as the global impact, 
% users with dense connections can spread influence to many other places~\cite{wc}.

\paragraph{Communal Attribute Voting (CAV).} 
Based on the definition above, we propose CAV, 
a graph augmentation method based on attribute masking, 
which adopts community voting for attribute removal:
% Community strength represents the ``loudness'' of each community here:
{\bf attributes which are more influential in strong communities are better preserved, 
whereas less influential attributes are more likely to be voted out.}

As the voting begins, a community penalty $w_a$ is assigned to every single attribute of $\mathbf{X}$. 
Intuitively, $w_a$ reflects the odds to be removed for every attribute candidate:
\begin{equation}\label{eq4}
    \mathbf{w}_a = \tilde{n}_a \left(\text{abs}(\mathbf{X})^\top\mathbf{H}\mathcal{S}\right)
\end{equation}
where $\mathbf{H} \in \{0, 1\}^{n \times |\mathcal{C}|}$ is an indicator matrix 
with $H_{i, c}=\mathbbm{1}_{[v_i\in c]}$ indicating to which community the $i$th node belongs,
and $\tilde{n}_a(x) = (x_{max} - x) / (x_{max} - x_{mean})$ is a one-dimensional normalization operation.
% Each node will vote ``softly'' for its redundant attributes,
Each node will score and vote for its redundant attributes,
and the least-voted-for attributes will win a greater opportunity to be preserved. 

Then, inspired by~\cite{gca}, we adaptively adjust the attribute perturbation distribution using $\mathbf{w}_a$. 
Specifically, the voting results are a group of attribute masks $\mathbf{m}_a$ 
independently sampled from two Bernoulli distributions:
\begin{equation}\label{eq5}
\resizebox{0.91\columnwidth}{!}{$
    \mathbf{m}_{a}^{(1)}\!\sim\!Bernoulli\!\left(1\!-\!\mathbf{w}_a p_{a}^{(1)}\!\right), \
    \mathbf{m}_{a}^{(2)}\!\sim\!Bernoulli\!\left(1\!-\!\mathbf{w}_a p_{a}^{(2)}\!\right)$}
\end{equation}
where $p_{a}^{(1)}$ and $p_{a}^{(2)}$ are two hyperparameters controlling the sampling ranges. 
In this way, community strength is well preserved in the perturbed attribute matrices:
\begin{equation}\label{eq6}
    \tilde{\mathbf{X}}^{(1)} = \mathbf{m}_a^{(1)} \circ \mathbf{X}, \
    \tilde{\mathbf{X}}^{(2)} = \mathbf{m}_a^{(2)} \circ \mathbf{X},
\end{equation}
where $\circ$ is the Hadamard product.

\paragraph{Communal Edge Dropping (CED).} 
The edge is the fundamental structural unit of graphs. 
Evolved from DropEdge~\cite{dropedge}, CED preserves community structures from perturbations guided by community strength.
% It prefers to {\it isolate communities}. 

The rationale for CED is that 
% we expect the perturbed graph to be aware of community strength: 
(1) {\bf intra-community edges are more important than inter-community edges}, and 
(2) {\bf edges in strong communities are more important than those in weak communities}. 
If there is a scoring function $w(e)$ to calculate the weight $w_e$ of each edge $e \in \mathcal{E}_c$ in the community $c$, 
it must meet the following condition:
\begin{equation}\label{eq7}
    w_e = w(e), \ s.t. \ w(e_{strong}) > w(e_{weak}) > w(e_{inter})
\end{equation}
where $e_{strong}$ is an intra-strong-community edge, 
$e_{weak}$ is an intra-weak-community edge with $\mathcal{S}_{strong} > \mathcal{S}_{weak}$, 
and $e_{inter}$ is an inter-community edge.
% Let $\mathbf{A}^{c} \in \{0,1\}^{|c|\times|c|}$ be the adjacency matrix of community $c$, 
% which is a principle submatrix of $\mathbf{A}$. 
% If the edge weight vector $\mathbf{w}_{e} = [w_{uv}]_{|\mathcal{E}|}$ for every $(u,v)\in\mathcal{E}$ is generated from an ideal scoring function $w(\mathcal{S})$,
% \begin{equation}\label{eq7}
%     \mathbf{w}_{e} = w(\mathcal{S}), \ s.t. \ w_{u_i u_j} > w_{v_i v_j} > w_{x_i x_j}, \nonumber
% \end{equation}
% \begin{subnumcases}{}
%     \forall a \in \mathcal{C}, \ \exists (u_i, u_j) \ s.t. \ A^{a}_{i, j}=1 \label{eq7a} \\
%     \forall b \in \mathcal{C}, \mathcal{S}_{a} > \mathcal{S}_{b}, \ \exists (v_i, v_j) \ s.t. \ A^{b}_{i, j}=1 \label{eq7b} \\ 
%     \forall (x_i, x_j), A_{i, j}=1, \ s.t. \ \forall c \in \mathcal{C}, A^{c}_{i, j}=0 \label{eq7c}
% \end{subnumcases}
% where $(u_i, u_j)$ in \eqref{eq7a} is an intra-strong-community edge, 
% $(v_i, v_j)$ in \eqref{eq7b} is an intra-weak-community edge, 
% and $(x_i, x_j)$ in \eqref{eq7c} is an inter-community edge. 

To this end, let $e=(u_i, u_j)$, we formulate CED as:
\begin{linenomath}
\begin{gather}\label{eq8}
    w_{e} = \tilde{n}_e \left(w(e)\right), \\
    \label{eq9}
    w(e) = 
    \begin{cases}
    \mathbf{H}_{i:}\mathcal{S}, 
    \phantom{--}
    (\mathbf{H}\mathbf{H}^\top \circ \mathbf{A})_{i,j}=1 \\
    -\left(\mathbf{H}_{i:}\mathcal{S} + \mathbf{H}_{j:}\mathcal{S}\right), \ \text{otherwise}
    \end{cases}
\end{gather}
\end{linenomath}
where $\tilde{n}_e(x) = (x - x_{min}) / (x_{mean} - x_{min})$ is a one-dimensional normalization operation.

\begin{theorem}
    $\tilde{n}_e \left(w(e)\right)$ defined in \eqref{eq8}-\eqref{eq9} satisfies \eqref{eq7}. 
\label{theorem1}
\end{theorem}
\begin{proof}
See Appendix~\ref{B}.
\end{proof}
\begin{rem}
Theorem 1 indicates that CED is an ideal graph augmentation on edges:
(1) strong communities have larger edge weights and vice versa; 
(2) inter-community edges have the smallest weights, which are much more likely to be dropped. 
Such design aims to retain real-world community properties in the perturbed graphs as much as possible. 
\end{rem}

Then, similar to CAV, the edge weight vector $\mathbf{w}_e$ is used to adjust the edge perturbation distribution: 
\begin{equation}\label{eq10}
\resizebox{0.89\columnwidth}{!}{$
    \mathbf{m}_{e}^{(1)}\!\sim\!Bernoulli \!\left(\mathbf{w}_e p_{e}^{(1)}\!\right),
    \mathbf{m}_{e}^{(2)}\!\sim\!Bernoulli \!\left(\mathbf{w}_e p_{e}^{(2)}\!\right)
$}
\end{equation}
where $p_{e}^{(1)}$ and $p_{e}^{(2)}$ function like $p_{a}^{(1)}$ and $p_{a}^{(2)}$. 
In this way, community strength is well preserved in the perturbed adjacency matrices:
\begin{equation}\label{eq11}
    \tilde{\mathbf{A}}^{(1)} = \left[
    m_{e, (u, v)}^{(1)} A_{(u, v)}^{(1)}
    \right], \ 
    \tilde{\mathbf{A}}^{(2)} = \left[
    m_{e, (u, v)}^{(2)} A_{(u, v)}^{(2)}
    \right]
\end{equation}

Up to now, two augmented graphs $(\tilde{\mathcal{G}}^{(1)}, \tilde{\mathcal{G}}^{(2)})$ 
have been generated by CSGCL to get embeddings.

\subsection{Team-up Contrastive Learning Scheme}

After obtaining the embeddings of two augmented graph views 
$\mathbf{Z}^{(1)}=f(\tilde{\mathcal{G}}^{(1)};\Theta),\mathbf{Z}^{(2)}=f(\tilde{\mathcal{G}}^{(2)};\Theta)$,
we find it necessary to teach our encoder the differences among communities, 
guided by a community-strength-enhanced objective. 

GRL with communities is analogous to real-world social group activities. 
On such occasions, people tend to team up with those with whom they are familiar. 
A group of strangers need to communicate with one another to seek common interests before they coalesce~\cite{p2pc}.
Going back to the GRL scenario:
at the beginning of training, nodes are unfamiliar with each other and crave for interaction with anyone within reach. 
% During this period, underscoring community structures is not helping, but imposing limits on the message passing process. 
As the training goes on, sufficient information exchange between nodes makes them understand their partners better. 
Driven by common interests, they tend to gather into groups -- 
at this moment, the guidance of communities is indispensable for model optimization. 

\begin{table*}[ht]
\centering
\captionsetup{justification=centering}
\resizebox{1.5\columnwidth}{!}{
\begin{tabular}{lcccccccc} \toprule 
Method & Training data & Level & Wiki-CS & Computers & Photo & Coauthor-CS  \\ \midrule
Raw (LogReg) & $\mathbf{X}$ & -- & 71.85±0.00 & 73.25±0.00 & 79.02±0.00 & 89.64±0.00\\ 
DeepWalk (w/o $\mathbf{X}$) & $\mathbf{A}$ & node & 73.84±0.16 & 85.77±0.58 & 89.06±0.43 & 84.71±0.35 \\ 
node2vec & $\mathbf{A}$ & node & 75.52±0.17 & 86.19±0.26 & 88.86±0.43 & 86.27±0.22\\ 
DeepWalk (w/ $\mathbf{X}$) & $\mathbf{X},\mathbf{A}$ & node & 77.21±0.03 & 86.28±0.07 & 90.05±0.08 & 87.70±0.04 \\ 
CPNE & $\mathbf{A}$ & community & 65.53±0.58 & 73.66±0.83 & 82.39±1.23 & OOM \\
\midrule\midrule
GAE & $\mathbf{X},\mathbf{A}$ & node & 70.15±0.01 & 85.27±0.19 & 91.62±0.13 & 90.01±0.71\\
VGAE & $\mathbf{X},\mathbf{A}$ & node & 75.63±0.19 & 86.37±0.21 & 92.20±0.11 & 92.11±0.09\\
DGI & $\mathbf{X},\mathbf{A}$ & node & 75.35±0.14 & 83.95±0.47 & 91.61±0.22 & 92.15±0.63\\
MVGRL & $\mathbf{X},\mathbf{A}$ & node & 77.52±0.08 & 87.52±0.11 & 91.74±0.07 & 92.11±0.12\\
GRACE & $\mathbf{X},\mathbf{A}$ & node & 77.68±0.34 & 88.29±0.11 & 92.52±0.34 & 92.50±0.08\\
GCA-best & $\mathbf{X},\mathbf{A}$ & node & 78.20±0.04 & 87.99±0.13 & 92.06±0.27 & 92.81±0.19\\
% GCA-best* & $\mathbf{X},\mathbf{A}$ & node & 78.35±0.05 & 87.85±0.31 & 92.53±0.16 & 93.10±0.01\\
AFGRL & $\mathbf{X},\mathbf{A}$ & node & 77.62±0.49 & 89.88±0.33 & 93.22±0.28 & 93.27±0.17\\
gCooL-best & $\mathbf{X},\mathbf{A}$ & community & 78.20±0.09 & 88.67±0.10 & 92.84±0.20 & 92.75±0.01\\
%COSTA & $\mathbf{X},\mathbf{A}$ & 79.12±0.02 & 88.32±0.03 & 92.56±0.45 & 92.94±0.10\\
%ProGCL & $\mathbf{X},\mathbf{A}$ & 78.45±0.04 & 89.55±0.16 & 93.64±0.13 & 93.67±0.12\\
CSGCL (ours) & $\mathbf{X},\mathbf{A}$ & community & \textbf{78.60±0.13} & \textbf{90.17±0.17} & \textbf{93.32±0.21} & \textbf{93.55±0.12}\\
\midrule\midrule
GCN & $\mathbf{X},\mathbf{A},\mathbf{Y}$ & -- & 78.02±0.51 & 87.79±0.36 & 91.82±0.01 & 93.06±0.00\\
GAT & $\mathbf{X},\mathbf{A},\mathbf{Y}$ & -- & 77.62±0.69 & 88.64±0.63 & 92.16±0.47 & 91.49±0.30\\ \bottomrule
\end{tabular}
}
\caption{Node classification results measured by accuracy (\%). 
``OOM'' stands for Out-Of-Memory on an 11GB GPU.
}
\label{cls}
\end{table*}

\begin{table*}[ht]
\centering
\begin{minipage}{0.95\columnwidth}
\captionsetup{justification=centering}
\tabcolsep=0.15cm
\resizebox{1\columnwidth}{!}{
\begin{tabular}{lcccc} \toprule
Method & Wiki-CS & Computers & Photo & Coauthor-CS  \\ \midrule
Raw (K-means) & 18.22±0.00 & 16.59±0.00 & 28.22±0.00 & 64.18±0.00  \\ 
CPNE & 32.44±1.30 & 42.51±1.60 & 53.62±1.58 & OOM \\ 
\hline \hline 
DGI & 31.00±0.02 & 31.80±0.02 & 37.60±0.03 & 74.70±0.01\\
MVGRL & 26.30±1.00 & 24.40±0.00 & 34.40±4.00 & 74.00±1.00\\
GRACE & 28.68±1.18 & 38.97±0.91& 47.70±4.28 & 73.79±0.58\\
GCA-best & 30.22±1.09 & 39.09±0.55 & 51.37±4.15 & 74.38±0.42\\
gCooL-best & 30.92±1.12 & 42.70±0.96 & 58.50±2.98 & 71.42±1.16\\
CSGCL (ours) & \textbf{32.80±1.50} & \textbf{45.09±0.95} & \textbf{58.79±2.17} & \textbf{78.29±1.24}\\ \bottomrule
\end{tabular}
}
\caption{
Node clustering results measured by NMI (\%).
}
\label{clu}
\end{minipage}
\quad\quad
\begin{minipage}{0.82\columnwidth}
\tabcolsep=0.15cm
\resizebox{1\columnwidth}{!}{
\begin{tabular}{lccc} \toprule 
Method & Computers & Photo & Coauthor-CS  \\ \midrule
Raw & 54.58±0.00 & 60.21±0.00 & 93.69±0.00\\
GRACE & 89.97±0.25& 88.64±1.17 & 87.67±0.10\\
GCA-best & 90.67±0.30 & 89.61±1.46 & 88.05±0.00\\
gCooL-best & 90.45±0.86 & 90.83±2.05& 88.91±0.04\\
CSGCL (ours) & \textbf{94.95±1.70} & \textbf{91.63±1.37} & \textbf{96.45±0.15}\\ \midrule\midrule
GCN & 87.89±0.90 & 88.20±0.08 & 92.71±0.63\\
GAT & 87.60±2.23 & 89.32±2.06 & 92.80±0.75\\ \bottomrule
\end{tabular}
}
\caption{
Link prediction results measured by AUC (\%).
}
\label{link}
\end{minipage}
\end{table*}

\begin{figure*}[!h]
  \centering
  \includegraphics[scale=0.23]{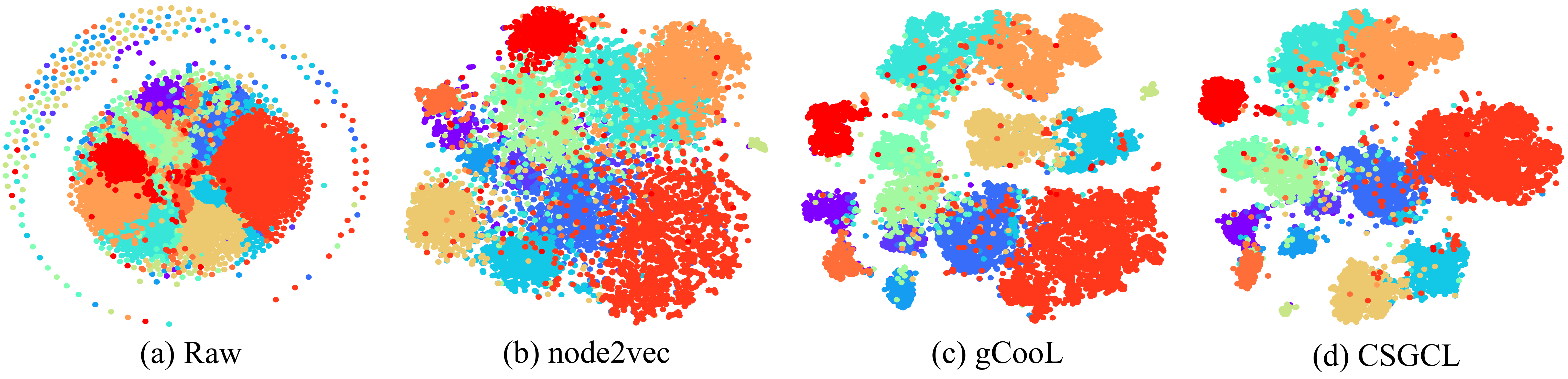}
  \caption{t-SNE visualization of representations on Coauthor-CS.}
  \label{fig-tsne}
\end{figure*}

Inspired by this, we propose the {\bf dynamic Team-up contrastive learning scheme}
to learn communities more effectively. 
We divide CSGCL's learning process into two phases:

\paragraph{Individual contrast.} 
In the first phase, 
the model optimizes InfoNCE
to learn initial similarities between each node individuals. 
The form of InfoNCE is shown in~\eqref{eq2}.
% This period is literally identical with the normal contrastive training procedure. 
% It is called ``individual contrast'' in correspondence with the Team-up contrast below. 

\paragraph{Team-up contrast.} 
In the second phase, every similarity
term $s$ of InfoNCE is fine-tuned by community strength: 
\begin{linenomath}
\begin{align}\label{eq12}
    &\mathcal{L} = 
    \mathbb{E}_{(\tilde{\mathbf{Z}}^{(1)}, \tilde{\mathbf{Z}}^{(2)})} \nonumber \\
    &\left(\!-\frac{1}{n}\!\sum_{i=1}^{n}{
    \log\!{\displaystyle\frac{\exp(\tilde{s}_{ii}^{(1,2)})}
    {\sum_{\substack{j=1 \\j \ne i}}^{n}\exp(\tilde{s}_{ij}^{(1,1)})\!+\!\sum_{j=1}^{n}\exp(\tilde{s}_{ij}^{(1,2)})}}
    }\!\right)
\end{align}
\end{linenomath}
where
\begin{equation}\label{eq13}
    \tilde{s}_{ij}^{(1,2)}=s_{ij}^{(1,2)}+\gamma(\mathbf{H}_{i:}+\mathbf{H}_{j:})\mathcal{S}.
\end{equation}
and $\mathbf{H}_{i:}$ is the community membership for the $i$th node.
The latter term in \eqref{eq13} refers to the strength of ``teams'' (communities) to which the $i$th and $j$th nodes belong,
and $\gamma$ is a coefficient of community strength. 
% A simple weighted sum will do the trick, 
% but the value of $\gamma$ will affect the model's judgment of the influence of each community. 

The final loss is a piecewise function combining the two phases above,
with a dynamic change of the fine-tuned similarity $\tilde{s}$ in \eqref{eq13}:
\begin{equation}\label{eq14}
    \tilde{s}_{ij}^{(1,2)}=s_{ij}^{(1,2)}+ \gamma(t)(\mathbf{H}_{i:}+\mathbf{H}_{j:})\mathcal{S}
\end{equation}
in which $\gamma(t)$ is a monotonically non-decreasing function that varies with training time $t$ (in the units of 100 epochs). 
We simply consider a hard Sigmoid-shaped form for $\gamma(t)$:
\begin{equation}\label{eq15}
    \gamma(t; t_0, \gamma_{max}) = \min\left\{\max\left\{0,t - t_0\right\}, \gamma_{max}\right\}
\end{equation}
where $t_0$ is the demarcation point of two phases. Thus we can unify two phases and formulate the final loss as \eqref{eq12}\eqref{eq14}\eqref{eq15}:
during the individual contrast ($t \le t_0$, $\gamma=0$), 
% free ``ice-breaking communications'' are allowed between node individuals. 
free communications are allowed between node individuals;
during the Team-up contrast ($t_0 < t < t_0 + \gamma_{max}$)
when the node-level relationships are well-learned,
$\gamma$ rises gradually to direct the model to notice community strength;
when teaming-up is complete ($t \ge t_0 + \gamma_{max}$), 
we set $\gamma\equiv\gamma_{max}$
to prevent deviation from the contrastive objective.

\section{Experiments}

In this section, we describe our experiments that were conducted to evaluate our model and answer the questions below:

\begin{itemize}  
\item Does GCL really benefit from our proposed methods? (Section~\ref{4.2}). 

\item Is the boost to performance really given by community strength? (Section~\ref{4.3}).
\item How does the Team-up strength coefficient influence the performance on a certain graph? (Section~\ref{4.4}).

\end{itemize}

Detailed experiment configurations can be found in Appendix~\ref{C}.
Additional experiment results using different community detectors and metrics (micro- \& macro-F1 and average precision)
can be found in Appendix~\ref{E}.

\subsection{Experiment Setup}\label{4.1}

\paragraph{Datasets.}
We use four benchmark graphs in different fields,
including one directed graph: Wiki-CS;
and three undirected graphs: Amazon-Computers (Computers), Amazon-Photo (Photo), and Coauthor-CS.
We convert Wiki-CS to an undirected graph only during the community detection process by adding reverse edges.
There is no other difference between Wiki-CS and undirected graphs for model training.
% See Appendix~\ref{C.1}. for more statistics.

\paragraph{Baselines.}
We divide all baseline models into the following three categories: 
\begin{itemize}
\item Traditional unsupervised models: 
Logistic regression (LogReg) \& K-means~\cite{kmeans} with raw features, DeepWalk~\cite{deepwalk} w/ or w/o the use of node attributes, 
node2vec~\cite{node2vec},
and CPNE~\cite{cpne}; 
\item Supervised graph neural network models: GCN~\cite{gcn} and GAT~\cite{gat};
\item Renowned and up-to-date self-supervised GRL models:
GAE \& VGAE~\cite{gae+vgae}, DGI~\cite{dgi}, 
MVGRL~\cite{mvgrl}, GRACE~\cite{grace}, 
GCA~\cite{gca}, AFGRL~\cite{afgrl}, and gCooL~\cite{gcool}. 
\end{itemize}
For GCA and gCooL which have multiple model configurations, 
we select the best one for each dataset,
marked as ``-best'' in the statistics.

% \paragraph{Implementation details.}
% See Appendix~\ref{C.2}.

\paragraph{Evaluation protocol.}
% For each experiment, the contrastive learning framework is trained unsupervisedly.
% Then, the representations are sent to different downstream branches. 
For node classification, we follow the evaluation protocol of~\cite{gca}
which trains and tests an $\ell_2$-regularized logistic regression classifier
with 10 random data splits (20 fixed splits for Wiki-CS). 
For node clustering, a K-means model~\cite{kmeans} with fixed initial clusters is fit.
For link prediction, the cosine similarity between every two nodes is calculated to evaluate the existence of a link.
Each experiment is repeated 10 times to report the average performance along with the standard deviation. 

\paragraph{Metrics.}
We use accuracy for node classification, 
normalized mutual information (NMI)~\cite{nmi} for node clustering, 
and area under the curve (AUC)~\cite{auc} for link prediction.
% The larger the values, the better the model's performance.

\paragraph{Implementation details.}
We choose Leiden~\cite{leiden} as the community detector for CSGCL,
before which we pretested the community detection methods detailed in Appendix~\ref{E.1}.
To ensure fairness,
A simple two-layer GCN is employed to CSGCL as well as all other contrastive baselines. 
We use the Adam optimizer to optimize the model. 
Detailed environment configurations can be found in Appendix~\ref{C.2}.

\begin{figure*}
  \centering
  \includegraphics[scale=0.84]{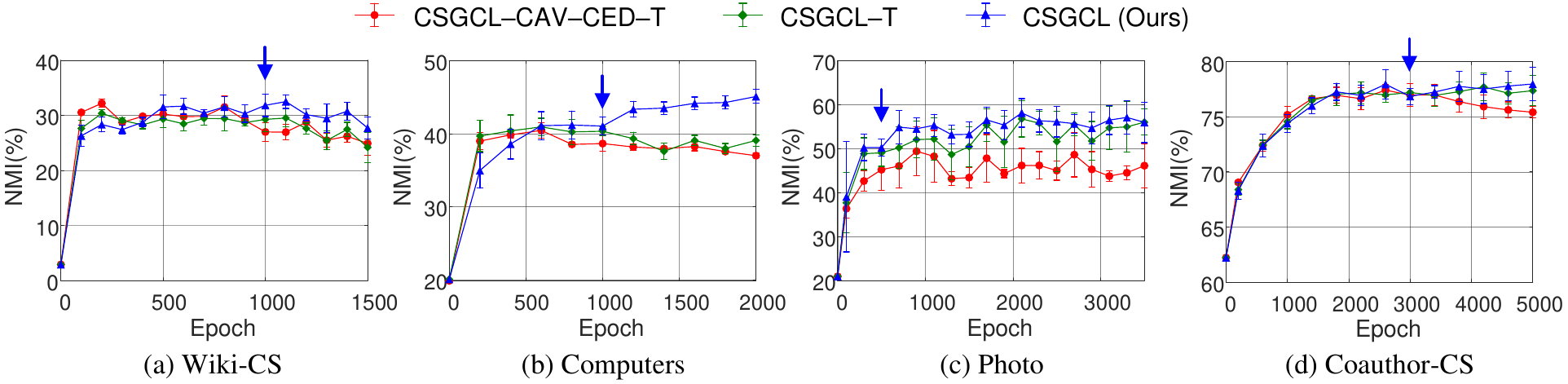}
  \caption{NMI (\%) of community detection in the ablation configurations.
  % Left to right: Wiki-CS, Computers, Photo, Coauthor-CS;
  ``\textcolor{blue}{$\downarrow$}'' points out $t_0$, the demarcation point of Individual and Team-up phases,
  which varies for different datasets.
  }
  \label{fig-ablation}
\end{figure*}

% \subsection{Competency and generality of CSGCL}\label{4.2}
\subsection{Overall Performance}\label{4.2}
In this section, we discuss the overall performance of CSGCL. 
CSGCL is adapted to three downstream tasks: 
node classification and node clustering on all datasets (Table~\ref{cls} and~\ref{clu}), and link prediction on three undirected datasets (Table~\ref{link}).
% Our undirected link prediction is not performed on Wiki-CS
% since it is not able to extract rich information from directed graphs~\cite{lpsurvey}.
Some statistics are borrowed from either their original papers or~\cite{gca,gcool}.
Data components used by each method during training are shown in the ``Training data'' column of Table~\ref{cls},
including node attributes $\mathbf{X}$, the adjacency matrix $\mathbf{A}$, and labels $\mathbf{Y}$.
The representation level of each unsupervised method is listed in the ``Level'' column,
including node-level and community-level.
{\bf Bolded} results in Tables~\ref{cls}--\ref{ablation} below represent the best results for each column.

We can see that 
CSGCL is superior to not only the traditional baselines 
but also the latest GCL models in all three downstream tasks. 
This is most noticeable in node clustering, 
where the NMI score of CSGCL is 1.7\% higher on average than the next-best methods; 
for link prediction, 
CSGCL has a 7.5\% increase of AUC on Coauthor-CS 
compared with other contrastive methods. 
CSGCL is also observed to be competitive against fully supervised models,
especially on link prediction.
% because supervised models are dependent on classification labels 
% that can not be applied directly to link prediction~\cite{sail}.
% Since CSGCL utilizes an undirected community detection algorithm,
% community semantics underlying in the directed graph, such as Wiki-CS, is damaged~\cite{dircd}.
% However, for link prediction on Wiki-CS, 
% CSGCL is still competitive with GCA and gCooL
% and outperforms other baselines.
Such achievements reveal the great generality of graph representations of CSGCL:
{\bf community strength is beneficial to multiple downstream tasks on graphs.}

Note that gCooL -- another GCL method with communities -- also achieves great performance on these tasks, 
which backs up the benefit of community semantics for GCL. 
However, 
% gCooL does not preserve community information throughout the learning process, 
% so their community information is perturbed during augmentation. 
gCooL does not capture the differences among communities. 
% either.
We conduct the one-tailed Wilcoxon signed-rank test~\cite{wilcoxon} to further back up our improvement:
taking node classification as an example, we have $p=9.77e$-$4<0.05$ on all datasets,
indicating the acceptance of alternative hypothesis that CSGCL has significant improvements over the best model of gCooL.

% Therefore, the accuracy of their best model 
% % (w/ exponential cosine similarity or Gaussian RBF similarity) 
% is about 0.5\%-1.5\% lower than ours. 
% % is 1.5\% less than ours on Computers.

Furthermore, we visualize the graph representations of CSGCL as well as baselines by t-SNE~\cite{tsne},
a dimension reduction and visualization method.
% to highlight inherent similarity among node samples. 
As shown in Figure~\ref{fig-tsne}, 
CSGCL learns the most discriminative graph representations among the competitive methods.

% \subsection{Validity of community-strength-enhancement}\label{4.3}
\subsection{Ablation Studies}\label{4.3}
In this section, we describe the ablation studies on the key modules of CSGCL. 
``--CAV'' and ``--CED'' refers to uniform attribute masking and edge dropping respectively,
where the perturbation probabilities are set to the same for all attributes and edges.
``--T'' refers to the InfoNCE objective instead of Team-up.
``--$\mathcal{S}$'' refers to the disregard of community strength,
where every community shares the average strength $\mathcal{S}$.

\begin{table}
\centering
\captionsetup{justification=centering}
\tabcolsep=0.08cm
\resizebox{1\columnwidth}{!}{
\begin{tabular}{lcccc} 
\toprule
Variant & Wiki-CS & Computers & Photo & Coauthor-CS  \\ \midrule
CSGCL--CAV--CED--T & 77.68±0.34 & 88.29±0.11 & 92.52±0.34 & 92.50±0.08 \\ 
CSGCL--CED--T & 77.78±0.07 & 89.94±0.48 & 93.08±0.28 & 93.35±0.13 \\ 
CSGCL--CAV--T & 78.43±0.11 & 90.06±0.19 & 93.11±0.34 & 93.09±0.11 \\ 
CSGCL--T & 78.51±0.11 & 90.12±0.23 & 93.26±0.23 & 93.50±0.09 \\
CSGCL-$\mathcal{S}$ & 77.88±0.15 & 89.52±0.26 & 92.85±0.34 & 93.41±0.07 \\
CSGCL (Ours) & \textbf{78.60±0.13} & \textbf{90.17±0.17} & \textbf{93.32±0.21} & \textbf{93.55±0.12} \\ \bottomrule
\end{tabular}}
\caption{Ablation study on node classification.}
\label{ablation}
\end{table}

Table~\ref{ablation} verifies the effectiveness of CSGCL.
The classification accuracy over all 4 datasets increases with either CAV or CED, 
and gains over 1\% improvement on average with both.
CSGCL with Team-up objective reaches the best classification performance,
bringing significant improvements over CSGCL--T on three datasets with the results $p=1.37e$-$2$, $2.78e$-$1$, $3.71e$-$2$, and $1.37e$-$2$ of Wilcoxon signed-rank test in order. 
More importantly, the accuracy degrades without variations of community strength $\mathcal{S}$.
This verifies that {\bf the difference among community strength cannot be ignored}.

We also show that CSGCL has well preserved community semantics 
by carrying out community detection on the representations.
Results are shown in Figure~\ref{fig-ablation}.
% We set different limits of training epochs for each dataset to ensure the results are not overfitted.
The green curves (w/ CAV and CED) perform better than the red curves (w/ uniform attribute masking and edge dropping) 
after a period of training time, especially on Photo and Coauthor-CS;
CSGCL, the blue curves using the Team-up objective, performs the best among its counterparts, 
with the performance gaps widening in the Team-up phase, 
especially on Computers and Photo. 
This shows that our community-strength-enhancement has better preserved community semantics throughout the learning process,
which results in more accurate graph representations.

% \subsection{Sensitivity of the strength coefficient}\label{4.4}
\subsection{Parameter Analysis}\label{4.4}

\begin{figure}
  \centering
  \includegraphics[scale=0.34]{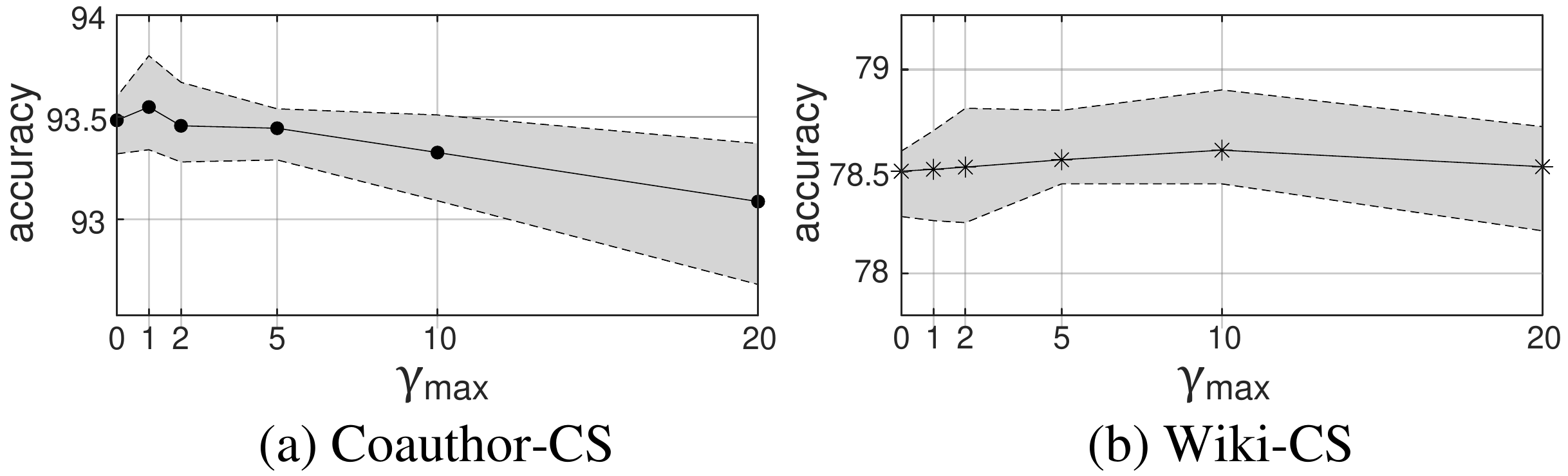}
  \caption{Node classification results with varied $\gamma_{max}$. 
  The average accuracy scores are dotted,
  and the shaded area is the error band. 
  % with broken lines delineating the upper \& lower bounds. 
  }
  \label{fig-loss}
\end{figure}

In this section, we describe the parameter analysis on the upper bound of the strength coefficient,
$\gamma_{max}$,
% Because of the progressive transition of the two phases,
% the variation of the demarcation point $t_0$ is not able to clearly show the effect of community strength.
% Therefore, we mainly discuss the effect of $\gamma_{max}$,
which controls the overall influence of communities on a certain graph.
All experiments below use the same training epochs and configurations.
$\gamma_{max}$ is set as 0, 1, 2, 5, 10, and 20.
% to enhance the readability of the experiment results.
We take Coauthor-CS and Wiki-CS as examples.

% Although the community semantics has different meanings for the two datasets,
As the results in Figure~\ref{fig-loss}, 
the accuracy of node classification is relatively stable when $\gamma_{max}$ changes in a certain range. 
So holistically, CSGCL is robust to the variation of $\gamma_{max}$, 
but it is recommended to use the most appropriate one for each dataset to achieve better performance.
The classification accuracy will decrease if $\gamma_{max}$ is too small ({\it e.g.} $< 1$ for Coauthor-CS),
undermining the benefit of communities;
however, it will also decrease if $\gamma_{max}$ is too large ({\it e.g.} $> 5$ for Coauthor-CS), 
derailing the model training. 

% \cite{gca} illustrated the effect of augmentation hyperparameters $p_a^{(1)}$, $p_a^{(2)}$, $p_e^{(1)}$, and $p_e^{(2)}$,
% showing the performance is relatively stable when they are not too large, 
% {\it e.g.}, $\le 0.7$.
% As the adaptivity of our augmentations is inspired by them, the conclusions above apply to CSGCL as well.

% for different graph datasets, the strength coefficient $\gamma_{max}$ varies to get the best results.

\section{Conclusion}

In this paper, we present CSGCL, 
a novel contrastive framework enhanced by community strength throughout the learning process.
Firstly, we manage to preserve differences among communities by the enhanced augmentations on attributes and edges, CAV and CED. 
Secondly, we put forward the dynamic Team-up contrastive scheme 
which regards GCL as a social group activity, 
guiding the optimization with community strength in a progressive manner.
CSGCL achieves state-of-the-art performance on three downstream tasks: node classification, node clustering, and link prediction,
indicating the effectiveness and generality of community-strength-enhanced representations.
% In the future, even though our community strength has well captured the differences among communities,
% we plan to study more precise definitions of community strength to further improve CSGCL.
% We will also investigate how community strength affects GCL on directed graphs.

\clearpage

\section*{Ethical Statement}

There are no ethical issues.

\section*{Acknowledgments}

This work is supported by National Natural Science Foundation of China under grants U1936108, U1836204, 62206102, 
and Science and Technology Support Program of Hubei Province under grant 2022BAA046.

%% The file named.bst is a bibliography style file for BibTeX 0.99c
\bibliographystyle{named}
\bibliography{ijcai23}

\clearpage

\appendix

\section{Algorithms of CSGCL}\label{A}

We summarize Communal Attribute Voting, Communal Edge Dropping and the whole CSGCL as Algorithm~\ref{alg:a}-\ref{alg} as follows.

\begin{algorithm}[htb]
    \small
    \caption{Communal Attribute Voting (CAV) algorithm.}
    \label{alg:a}
    \textbf{Input}: Attribute matrix $\mathbf{X}$, community indicator matrix $\mathbf{H}$, community strength vector $\mathcal{S}$\\
    \textbf{Parameter}: Sample probability $p_a$\\
    \textbf{Output}: Perturbed attribute matrix $\tilde{\mathbf{X}}$
    \begin{algorithmic}[1] %[1] enables line numbers
        \IF {all-zero dimension $\mathbf{X}_{:j}$ exists}
        \STATE {exclude $\mathbf{X}_{:j}$ from $\mathbf{X}$}
        \ENDIF
        \STATE Calculate normalized attribute penalty vector $\mathbf{w}_a$ by~\eqref{eq4}
        \STATE Sample attribute masks 
        $\mathbf{m}_a$ from $Bernoulli (1 - \mathbf{w}_a p_a)$
        \STATE $\tilde{\mathbf{X}} = \mathbf{m}_a \circ \mathbf{X}$
        \STATE \textbf{return} $\tilde{\mathbf{X}}$
    \end{algorithmic}
\end{algorithm}

\begin{algorithm}[htb]
    \small
    \caption{Communal Edge Dropping (CED) algorithm.}
    \label{alg:e}
    \textbf{Input}: Adjacency matrix $\mathbf{A}$, community indicator matrix $\mathbf{H}$, community strength vector $\mathcal{S}$\\
    \textbf{Parameter}: Sample probability $p_{e}$\\
    \textbf{Output}: Perturbed adjacency matrix $\tilde{\mathbf{A}}$
    \begin{algorithmic}[1] %[1] enables line numbers
        \FOR {$e \in \mathcal{E}$}
        \STATE Calculate edge weight $w(e)$ by~\eqref{eq9}
        \ENDFOR
        \STATE Calculate normalized edge weight vector $\mathbf{w}_{e} = \tilde{n}_e \left([w(e)]\right)$
        \STATE Sample edge masks 
        $\mathbf{m}_e$ from $Bernoulli (\mathbf{w}_e p_e)$
        \FOR {$(u, v) \in \mathcal{E}$}
        \STATE $\tilde{A}_{(u, v)} = m_{e, (u, v)} A_{(u, v)}$
        \ENDFOR
        \STATE \textbf{return} $\tilde{\mathbf{A}}$
    \end{algorithmic}
\end{algorithm}

\begin{algorithm}[!htb]
    \small
    \caption{CSGCL algorithm.}
    \label{alg}
    \textbf{Input}: Attribute graph $\mathcal{G} = (\mathbf{X}, \mathbf{A})$\\
    \textbf{Parameter}: Sample probabilities $p_a^{(1)}$, $p_a^{(2)}$, $p_e^{(1)}$, $p_e^{(2)}$, loss hyperparameters $t_0$, $\gamma_{max}$, $\tau$ \\
    \textbf{Output}: A trained encoder $f(\cdot;\Theta)$
    \begin{algorithmic}[1] %[1] enables line numbers
        \STATE Initialize $\Theta$
        \STATE Find communities $\mathbf{H}$ of $\mathcal{G}$
        \STATE Calculate community strength vector $\mathcal{S}$ by~\eqref{eq3}
        \FOR {epoch = $1,2,\cdots$}
        \STATE // 6-10: Generate two perturbed graph views
        \STATE $\tilde{\mathbf{X}}^{(1)} = \text{CAV}(\mathbf{X},\mathbf{H},\mathcal{S};p_a^{(1)})$
        \STATE $\tilde{\mathbf{A}}^{(1)} = \text{CED}(\mathbf{A},\mathbf{H},\mathcal{S};p_e^{(1)})$
        \STATE $\tilde{\mathbf{X}}^{(2)} = \text{CAV}(\mathbf{X},\mathbf{H},\mathcal{S};p_a^{(2)})$
        \STATE $\tilde{\mathbf{A}}^{(2)} = \text{CED}(\mathbf{A},\mathbf{H},\mathcal{S};p_e^{(2)})$
        \STATE $\tilde{\mathcal{G}}^{(1)} = (\tilde{\mathbf{X}}^{(1)}, \tilde{\mathbf{A}}^{(1)}), \tilde{\mathcal{G}}^{(2)} = (\tilde{\mathbf{X}}^{(2)}, \tilde{\mathbf{A}}^{(2)})$
        \STATE $\tilde{\mathbf{Z}}^{(1)} = f(\tilde{\mathcal{G}}^{(1)};\Theta), \tilde{\mathbf{Z}}^{(2)} = f(\tilde{\mathcal{G}}^{(2)};\Theta)$
        \STATE Calculate Team-up loss $\mathcal{L}(\mathbf{H}, \mathcal{S}, t_0, \gamma_{max}, \tau)$ by~\eqref{eq12}\eqref{eq14}\eqref{eq15}
        \STATE Update $\Theta$ by gradient descent on $\mathcal{L}$
        \ENDFOR
    \end{algorithmic}
\end{algorithm}

\section{Proof of Theorem~\ref{theorem1}}\label{B}

Suppose the community indicator matrix $\mathbf{H}=[H_{i,k}]_{n \times |\mathcal{C}|}$ 
where $H_{i,k} = 1$ iff the $i$th node belongs to the community $k$, else $H_{i,k} = 0$. 
Then we have
\begin{linenomath}\label{eq16}
\begin{align}\nonumber
(\mathbf{H}\mathbf{H}^\top \circ \mathbf{A})_{i,j} & = A_{i,j}(\mathbf{H}\mathbf{H}^\top)_{i,j} = A_{i,j}\mathbf{H}_{i:}\mathbf{H}_{j:}^\top \\
& = A_{i,j}\sum_{k=1}^{|\mathcal{C}|}{H_{i,k}H_{j,k}}
\end{align}
\end{linenomath}
Consider an intra-strong-community edge $e_{strong} = (u_i,u_j)$, and an intra-weak-community edge $e_{weak} = (u_k,u_l)$, 
then we have $H_{i,strong} = H_{j,strong} = H_{k,weak} = H_{l,weak} = 1$, $A_{i,j} = A_{k,l} = 1$. 
Therefore, 
\begin{linenomath}\label{eq17}
\begin{gather}
    (\mathbf{H}\mathbf{H}^\top \circ \mathbf{A})_{i,j} = A_{i,j}H_{i,strong}H_{j,strong} = A_{i,j} = 1 \\
    \label{eq18}
    (\mathbf{H}\mathbf{H}^\top \circ \mathbf{A})_{k,l} = A_{k,l}H_{k,weak}H_{l,weak} = A_{k,l} = 1
\end{gather}
\end{linenomath}
So according to~\eqref{eq9}:
\begin{equation}\label{eq19}
    \resizebox{0.88\columnwidth}{!}{$
    w(e_{strong})=\mathbf{H}_{i:}\mathcal{S}=\mathcal{S}_{strong}, \ 
    w(e_{weak})=\mathbf{H}_{j:}\mathcal{S}=\mathcal{S}_{weak}$
    }
\end{equation}
Since $\mathcal{S}_{strong} > \mathcal{S}_{weak}$, $w(e_{strong}) > w(e_{weak})$.

Consider an inter-community edge $e_{inter} = (u_m,u_n)$,
where $u_m$ and $u_n$ are not members of the same community:
$\nexists c \in \mathcal{C} \ s.t. \ H_{m,c}=H_{n,c}=1$.
Therefore, 
\begin{equation}\label{eq20}
    (\mathbf{H}\mathbf{H}^\top \circ \mathbf{A})_{m,n} = A_{m,n}\underbrace{\sum_{k=1}^{|\mathcal{C}|}{H_{m,k}H_{n,k}}}_0 = 0
\end{equation}
Similarly, according to~\eqref{eq9}:
\begin{equation}\label{eq21}
    w(e_{inter}) = -(\mathbf{H}_{m:}\mathcal{S} + \mathbf{H}_{n:}\mathcal{S}) 
    < \min_{c \in \mathcal{C}}{\mathcal{S}_c}
    < \mathcal{S}_{weak}.
\end{equation}
We know from~\eqref{eq19} that $\mathcal{S}_{weak} = w(e_{weak}) > w(e_{inter})$, so we finally have $w(e_{strong}) > w(e_{weak}) > w(e_{inter})$.
As the normalization operation $\tilde{n}_e(\cdot)$ is linear with $x_{mean} > x_{min}$, 
it will not change our conclusion. 
Thus, $\tilde{n}_e \left(w(e)\right)$ defined in \eqref{eq8}-\eqref{eq9} satisfies the condition of~\eqref{eq7}.
\qed

\section{Experiment Details}\label{C}

\subsection{Datasets}\label{C.1}
% For the main experiments, 
We use Wiki-CS, Amazon-Computers (Computers), Amazon-Photo (Photo), and Coauthor-CS to conduct our experiments. 
The detailed statistics of all the used datasets are in Table~\ref{data}.

\begin{table}[hbt]
\centering
\small
\resizebox{1\columnwidth}{!}{
\tabcolsep=0.1cm
\begin{tabular}{llrrrr} \toprule 
Dataset & Type & Nodes & Edges & Attributes & Classes  \\ \midrule
Wiki-CS & reference & 11,701 & 216,123 & 300 & 10  \\ 
Photo & co-purchase & 7,487 & 119,043 & 745 & 8\\ 
Computers & co-purchase & 13,381 & 245,778 & 767 & 10\\ 
Coauthor-CS & co-author & 18,333 & 81,894 &	6,805 & 15\\
\bottomrule
\end{tabular}
}
\caption{Statistics of data.}
\label{data}
\end{table}

\begin{table*}[htb]
\centering
\small
\begin{tabular}{lccccccccccc} \toprule 
Dataset & \makecell{Training \\ epochs} & $p_a^{(1)}$ & $p_a^{(2)}$ & $p_e^{(1)}$ & $p_e^{(2)}$ & $t_0$ & $\gamma_{max}$ & $\tau$ & \makecell{Learning \\ rate} & \makecell{Hidden \\ dimension} & \makecell{Activation \\ function} \\ \midrule
Wiki-CS & 2000 & 0.1 & 0.2 & 0.2 & 0.7 & 10 & 10 & 0.6 & 1e-2 & 256 & PReLU \\ 
Computers & 2500 & 0.1 & 0.1 & 0.6 & 0.7 & 20 & 5 & 0.2 & 1e-1 & 512 & PReLU \\ 
Photo & 3000 & 0.1 & 0.1 & 0.5 & 0.7 & 20 & 1 & 0.3& 1e-1 & 256 & ReLU  \\ 
Coauthor-CS & 4600 & 0.1 & 0.5 & 0.1 & 0.1 & 15 & 1 & 0.5 & 5e-4 & 256 & RReLU \\
\bottomrule
\end{tabular}
% }
\caption{Detailed hyperparameters.}
\label{param}
\end{table*}

\begin{table*}[htb]
\centering
% \captionsetup{justification=centering}
\resizebox{1.55\columnwidth}{!}{
\begin{tabular}{llccccc} \toprule
Metric & Method & Level & Wiki-CS & Computers & Photo & Coauthor-CS  \\ \midrule
\multirow{5}{*}{Micro-F1} 
& Raw (LogReg) & -- & 71.86±0.00 & 73.19±0.00 & 79.02±0.00 & 89.64±0.00 \\
& GRACE & node & 78.00±0.17 & 88.28±0.15 & 92.62±0.14 & 92.50±0.01\\
& GCA-best & node & 78.19±0.05 & 87.99±0.15 & 92.41±0.20 & 92.92±0.06\\
& gCooL-best & community & \textbf{78.71±0.05} & 88.67±0.10 & 92.93±0.24 & 92.75±0.01\\
& CSGCL & community & 78.56±0.15 & \textbf{90.17±0.18} & \textbf{93.28±0.24} & \textbf{93.52±0.07}\\ \midrule\midrule
\multirow{5}{*}{Macro-F1} 
& Raw (LogReg) & -- & 67.98±0.00 & 55.51±0.00 & 74.07±0.00 & 85.13±0.00\\
& GRACE & node & 74.99±0.25 & 86.32±0.27 &91.39±0.36 &89.89±0.01\\
& GCA-best & node & 75.30±0.08 & 85.93±0.31 & 91.04±0.30 & 90.76±0.16\\
& gCooL-best & community & 75.89±0.09 & 87.38±0.14 & 91.81±0.26 & 90.87±0.03\\
& CSGCL & community & \textbf{75.90±0.12} & \textbf{89.02±0.20} & \textbf{91.96±0.15} & \textbf{91.84±0.11}\\ \bottomrule
\end{tabular}
}
\caption{Comparative results of node classification measured by F1 scores ($\%$).
}
\label{cls-f1}
\end{table*}

\begin{itemize}  
\item  {\bf Wiki-CS}~\cite{wikics} 
is a directed graph dataset obtained from Wikipedia, 
composed of Computer Science articles and hyperlinks between articles. 
There are 10 types of nodes, representing articles in different fields. 
Node attributes are calculated as the average value of the text embedding of corresponding articles. 
It is the only dataset with dense attributes. 
\item  {\bf Amazon-Computers} and {\bf Amazon-Photo}~\cite{amazon+coauthor} 
are two networks representing the co-purchase relations of goods. 
Edges indicate that two goods are purchased together more frequently. 
Node attributes are bag-of-words encoding product reviews. 
Classes are given by the categories of products. 
\item  {\bf Coauthor-CS}~\cite{amazon+coauthor} 
is a Microsoft academic network based on the 2016 KDD Cup Challenge. 
Each node represents an author of a paper, and its attributes represent the paper keywords. 
Two linked authors have collaborated on a single paper.
Each category tag indicates the most active research direction of each author. 
\end{itemize}

\begin{figure}
  \centering
  \includegraphics[scale=0.48]{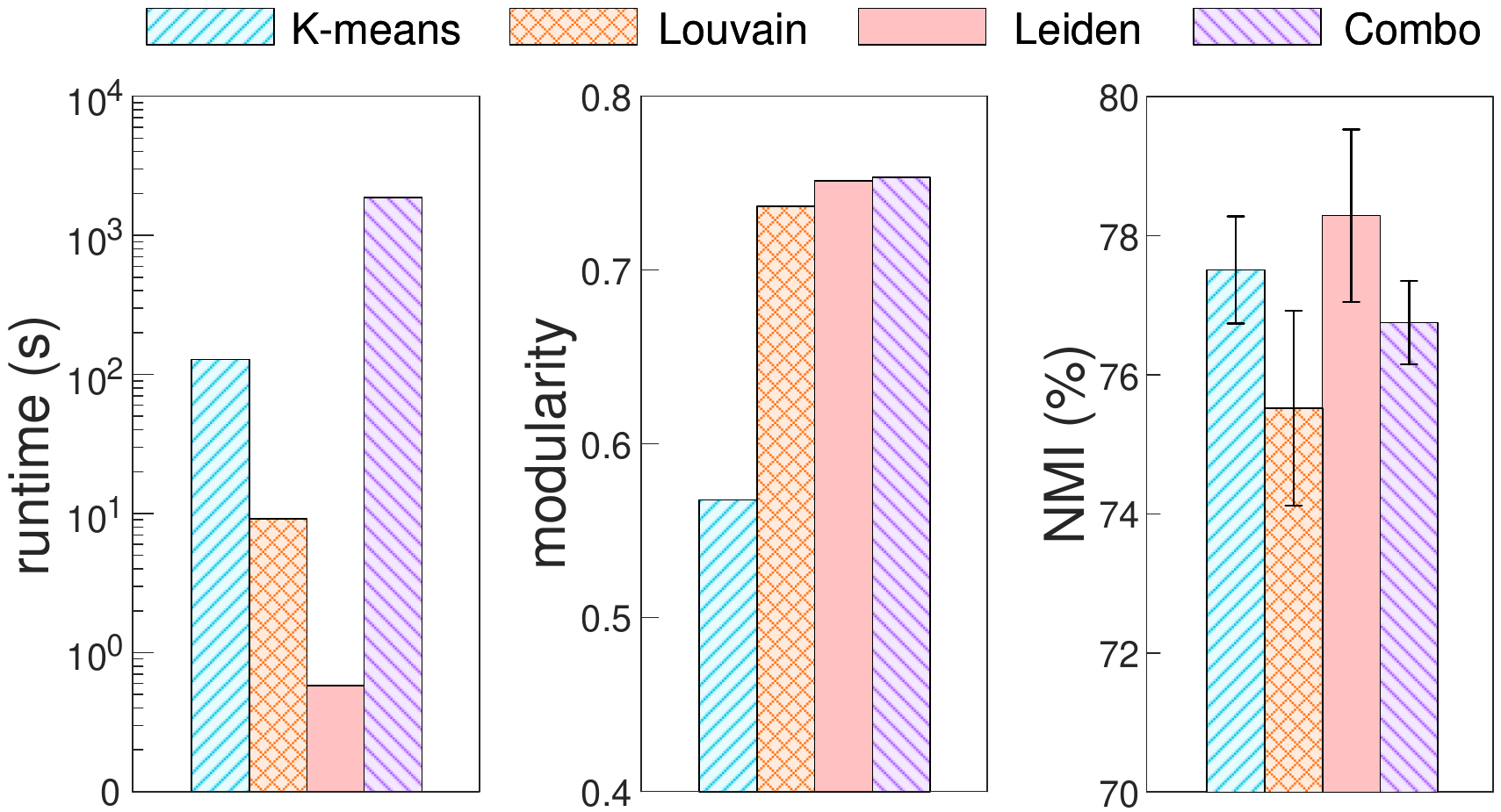}
  \caption{
  Performance comparison of different community detection methods on Coauthor-CS,
  including the average runtime of each algorithm (left),
  the average modularity of each partition (center),
  and performance of CSGCL on the node clustering task when each algorithm is employed as the community detector (right).
  }
  \label{fig-cd}
\end{figure}

\subsection{Environment Configurations}\label{C.2}
All experiments are conducted on an NVIDIA GeForce GTX 1080 Ti GPU (with 11GB memory).
The model is built using PyTorch 1.8.1~\cite{pytorch} and PyTorch Geometric 2.0.1~\cite{pyg},
the latter is also the source of all datasets. 
The community detection algorithm is provided by CDLib 0.2.6.\footnote{
\url{https://github.com/GiulioRossetti/cdlib}
}
We adopt grid search for the appropriate hyperparameters, following~\cite{gca}. 
The detailed hyperparameters for node classification are shown in Table~\ref{param}.

\section{Related Work of Community Detection}\label{D}

Community detection methods identify underlying communities from a large-scale network. 
Traditional unsupervised heuristic methods include graph partitioning~\cite{girvan-newman},
hierarchical agglomeration~\cite{hierarchy}, 
spectral clustering~\cite{spectral}, 
density-peak clustering~\cite{lccd}, 
and modularity optimization~\cite{mod}. 
Louvain~\cite{louvain} is one of the most popular optimization method which divides each iteration into two steps: 
(1) introduce a node perturbation into each community, \textit{i.e.} relocate a node to an adjacent community; 
(2) evaluate the gain of modularity as a result of the perturbation to decide if it will be retained. 
Leiden~\cite{leiden} further improves Louvain by refining the partition to get rid of disconnected subgraphs. 
Combo~\cite{combo} is another accurate optimization method 
that shifts every node from its source community to a better redistributed destination community for each iteration. 
However, Combo suffers from large time complexity.
Deep learning based community detection methods use graph convolutional networks (GCNs)~\cite{sgcn,clare}, 
generative adversarial networks (GANs)~\cite{seal}, 
or graph autoencoders (GAEs)~\cite{graphencoder} to detect communities under more complicated circumstances. 
For our model, we utilize unsupervised community detection above to precisely mine the community information, 
in order to refocus attention on the performance of communities in the graph representations. 

\section{Additional Experiments}\label{E}

\subsection{Performance of Different Community Detection Methods}\label{E.1}

CSGCL is compatible with multiple non-overlapping unsupervised community detection methods.
Different to the methods that need to specify the number of communities beforehand, 
CSGCL prefers other methods that do not need to do this, 
because they do not rely on extra prior knowledge of the training data,
thus they are suitable for new graphs without explicit community information.

To choose an appropriate community detection method, 
we compare three available candidates: 
{\bf Louvain}~\cite{louvain}, {\bf Leiden}~\cite{leiden}, and {\bf Combo}~\cite{combo}.
As a supplement, {\bf K-means}~\cite{kmeans}, a community number fixed method, 
is also tested with the number of communities $K$ set as the actual number of classes (15 for Coauthor-CS), as recommended by~\cite{gcool}.

We obtained the results of average runtime, modularity of community partition, 
and final performance of node clustering on Coauthor-CS, as shown in Figure~\ref{fig-cd}. 
We analyze the observations and draw the following conclusions from the results:
(1) Our model w/ K-means achieves the next-best performance in the node clustering task, 
but it is discarded because of its higher demand for data.
(2) Combo is the best community partition method in our experiment,
but it is dismissed because of its prohibitive time overhead on large-scale networks,
being orders of magnitude higher than other algorithms.
(3) Leiden is the fastest algorithm (with an average runtime of under 10 seconds) with the best performance on the node clustering task.

To sum up, Leiden is chosen as the community detector in all of our experiments.
More importantly, it is likely that our CSGCL will achieve more impressive performance 
with better community detection methods in the future. 

\begin{table}
\centering
\small
\resizebox{0.9\columnwidth}{!}{
\tabcolsep=0.1cm
\begin{tabular}{lccc} \toprule 
Method & Computers & Photo & Coauthor-CS  \\ \midrule
Raw & 85.42±0.00 & 83.14±0.00 & 86.59±0.00\\
GRACE & 92.15±0.43 & 83.85±4.15 & 94.87±0.02\\
GCA-best & 90.50±0.63 & 86.53±3.00 & 94.94±1.37\\
gCooL-best & 86.06±2.33 & 84.01±2.50 & 88.24±0.05\\
CSGCL (ours) & \textbf{94.86±1.04} & \textbf{89.45±1.74} & \textbf{96.00±0.14}\\ \midrule\midrule
GCN & 87.01±1.26 & 85.87±1.02 & 89.02±1.67\\
GAT & 86.00±4.75 & 89.18±2.71 & 90.90±1.56\\ \bottomrule
\end{tabular}
}
\caption{Link prediction measured by AP (\%).}
\label{link-ap}
\end{table}

\subsection{Performance by Different Metrics}\label{E.2}

We also tried micro-F1 and macro-F1, two metrics for multi-class classification, 
to better evaluate the performance of node classification. 
The results are presented in Table~\ref{cls-f1}.
{\bf Bolded} results represent the best results for each column.
Table~\ref{cls-f1} indicates that CSGCL is still dominant over node-level GCL methods
and outperforms gCooL on Computers, Photo, and Coauthor-CS
in both micro-F1 and macro-F1. 

Additionally, we use average precision (AP) as an additional metric of link prediction
to compensate for the biases of AUC.
Table~\ref{link-ap} shows that CSGCL is still dominant over the baselines in AP.

\end{document}